\newif\ifconf\conffalse

\ifconf
\documentclass[twoside,leqno,twocolumn]{article}
\usepackage{ltexpprt}

\else
\documentclass[11pt,letterpaper]{article}
\usepackage{amsthm}
\usepackage[margin=1in,dvips]{geometry}
\fi

\usepackage{color}
\usepackage{import}
\usepackage{verbatim}
\usepackage{tikz}
\usetikzlibrary{quotes,angles,decorations.pathreplacing,calc}
\usepackage{graphicx}
\usepackage{amssymb}
\usepackage{amsmath}
\usepackage{float}
\usepackage{enumerate}
\usepackage{kpfonts}
\usepackage[boxed,section]{algorithm}
\usepackage{algpseudocode}
\usepackage{multirow}
\usepackage[hidelinks]{hyperref}

\iffalse
\newtheorem{definition}{Definition}[section]
\newtheorem{remark}{Remark}[section]
\else
\newtheorem{theorem}{Theorem}[section]
\newtheorem{lemma}[theorem]{Lemma}

\newtheorem{definition}[theorem]{Definition}

\newtheorem{propn}[theorem]{Proposition}
\newtheorem{remark}[theorem]{Remark}

\fi

\newcommand{\abs}[1]{\left|#1\right|}

\newcommand{\norm}[2]{\left \lVert#2\right \rVert_{#1}}

\DeclareMathOperator{\disj}{\mathsf{MostlyDISJ}}
\newcommand{\dtv}{d_{\text{TV}}}
\newcommand{\noleftdelimiter}{\left.\kern-\nulldelimiterspace}

\let\emptyset\varnothing

{\makeatletter
	\gdef\xxxmark{%
		\expandafter\ifx\csname @mpargs\endcsname\relax % in minipage?
		\expandafter\ifx\csname @captype\endcsname\relax % in 
		\marginpar{xxx}% not in a caption or minipage, can use marginpar
		\else
		xxx % notice trailing space
		\fi
		\else
		xxx % notice trailing space
		\fi}
	\gdef\xxx{\@ifnextchar[\xxx@lab\xxx@nolab}
	\long\gdef\xxx@lab[#1]#2{{\bf [\xxxmark #2 ---{\sc #1}]}}
	\long\gdef\xxx@nolab#1{{\bf [\xxxmark #1]}}
}

\DeclareMathOperator{\supp}{supp}
\DeclareMathOperator*{\E}{\mathbb{E}}
\def\R{\mathbb{R}}

\def\eps{\varepsilon}
\def\wh{\widehat}
\def\wt{\widetilde}
\def\Thetat{\wt{\Theta}}
\def\cD{\mathcal{D}}

\usepackage{etoolbox}

\makeatletter

\newcommand{\define}[4][ignore]{%
  \ifstrequal{#1}{ignore}{}{
  \@namedef{thmtitle@#2}{#1}}%
  \@namedef{thm@#2}{#4}%
  \@namedef{thmtypen@#2}{lemma}%
  \newtheorem{thmtype@#2}[theorem]{#3}%
  \newtheorem*{thmtypealt@#2}{#3~\ref{#2}}%
}

\newcommand{\state}[1]{%
  \@namedef{curthm}{#1}
  \@ifundefined{thmtitle@#1}{
  \begin{thmtype@#1}
    }{
  \begin{thmtype@#1}[\@nameuse{thmtitle@#1}]
  }
    \label{#1}
    \@nameuse{thm@#1}
  \end{thmtype@#1}
  \@ifundefined{thmdone@#1}{
  \@namedef{thmdone@#1}{stated}%
  }{}
}

\newcommand{\restate}[1]{%
  \@namedef{curthm}{#1}
  \@ifundefined{thmtitle@#1}{
    \begin{thmtypealt@#1}
    }{
  \begin{thmtypealt@#1}[\@nameuse{thmtitle@#1}]
  }
    \@nameuse{thm@#1}
  \end{thmtypealt@#1}
  \@ifundefined{thmdone@#1}{
  \@namedef{thmdone@#1}{stated}%
  }{}
}

\newcommand{\thmlabel}[1]{
  \@ifundefined{thmdone@\@nameuse{curthm}}{\label{#1}
    }{\tag*{\eqref{#1}}}
}
\makeatother
\begin{document}
\title{A Simple Proof of a New Set Disjointness with Applications to Data Streams}
\author{
	Akshay Kamath\\
	UT Austin\\\texttt{kamath@cs.utexas.edu} 
	\and
	Eric Price\\
	UT Austin\\\texttt{ecprice@cs.utexas.edu}
	\and
	David P. Woodruff\\
	CMU\\\texttt{dwoodruf@cs.cmu.edu} }
\date{}
\begin{titlepage}
\maketitle
\thispagestyle{empty}

\begin{abstract}
  The multiplayer promise set disjointness is one of the most widely
  used problems from communication complexity in applications. 
  In this problem there are $k$ players with
  subsets $S^1, \ldots, S^k$, each drawn from $\{1, 2, \ldots, n\}$,
  and we are promised that either the sets are (1) pairwise disjoint, or
  (2) there is a unique element $j$ occurring in all the sets, which are 
  otherwise pairwise disjoint. The total communication 
  of solving this problem with constant probability 
   in the blackboard model is $\Omega(n/k)$. 

  We observe for most applications, it instead suffices to look at what we call the 
  ``mostly'' set disjointness problem, which changes case (2) to say
  there is a unique element $j$ occurring in at least half of the sets, and the
  sets are otherwise disjoint. This change gives us a much simpler proof
  of an $\Omega(n/k)$ randomized total communication lower bound, avoiding Hellinger
  distance and Poincare inequalities. Our proof also gives 
  strong lower bounds for high probability protocols, which are much larger than what
  is possible for the set disjointness problem. Using this we 
  show several new results for data streams:

  \begin{enumerate}
  \item for $\ell_2$-Heavy Hitters,
  any $O(1)$-pass streaming algorithm in the 
  insertion-only model for detecting if an $\eps$-$\ell_2$-heavy hitter exists 
  requires $\min(\frac{1}{\eps^2}\log \frac{\eps^2n}{\delta}, \frac{1}{\eps}n^{1/2})$ 
  bits of memory, which is optimal up to a $\log n$ factor. 
  For deterministic algorithms and constant $\eps$,
  this gives an $\Omega(n^{1/2})$ lower bound, improving the prior $\Omega(\log n)$ lower 
  bound. We also obtain lower bounds for Zipfian 
  distributions. 
  \item for $\ell_p$-Estimation, $p > 2$, we show an $O(1)$-pass 
        $\Omega(n^{1-2/p} \log(1/\delta))$ bit lower bound for outputting an $O(1)$-
        approximation with probability $1-\delta$, in the insertion-only model. This is optimal, 
        and the best previous lower bound was $\Omega(n^{1-2/p} + \log(1/\delta))$.           
   \item for low rank approximation of a sparse matrix in $\R^{d\times n}$, if 
   we see the rows of a matrix
         one at a time in the row-order model, each row having $O(1)$ non-zero 
         entries, any deterministic
         algorithm requires $\Omega(\sqrt{d})$ memory to output an $O(1)$-approximate rank-$1$ 
         approximation.
  \end{enumerate}
  Finally, we consider 
  strict and general turnstile streaming models, and show separations
  between sketching lower bounds and non-sketching upper bounds for the heavy 
  hitters problem.
\end{abstract}
\end{titlepage}

\section{Introduction}
Communication complexity is a common technique for establishing lower bounds
on the resources required of problems, such as the memory required of a 
streaming algorithm.  
The multiplayer promise set disjointness is one of the most widely
  used problems from communication complexity in applications, not only in
  data streams \cite{AMS99,BJKS,cks03,g09,j09a,j09,chan2011edit}, but also 
  compressed sensing \cite{p13}, 
  distributed functional monitoring \cite{WZ12,woodruff2014optimal}, distributed learning 
  \cite{g14,k14,b16},
  matrix-vector query
  models \cite{SWYZ19}, voting \cite{MP0W19,MSW20}, and so on. We shall restrict ourselves
  to the study of set disjointness in the number-in-hand communication model, 
  described below, which covers all of the above applications. 
  Set disjointness is also well-studied in the number-on-forehead communication 
  model, see, e.g., \cite{grolmusz1994bns,tesson2003computational,beame2006strong,lee2009disjointness,chattopadhyay2008multiparty,beame2012multiparty,sherstov2012multiparty,sherstov2014communication}, though we will not discuss that model here. 
 
  In the number-in-hand multiplayer promise set disjointness 
  problem there are $k$ players with
  subsets $S^1, \ldots, S^k$, each drawn from $\{1, 2, \ldots, n\}$,
  and we are promised that either:
  \begin{enumerate}
  \item the $S^i$ are pairwise disjoint, or
  \item there is a unique element $j$ occurring in all the sets, which are 
  otherwise pairwise disjoint. 
  \end{enumerate}
  The promise set disjointness problem was posed by Alon, Matias, and Szegedy 
  \cite{AMS99}, who showed an $\Omega(n/k^4)$ total communication bound in the 
  blackboard communication model, where each player's message can be seen by 
  all other players. This total communication bound was then improved to 
  $\Omega(n/k^2)$ by Bar-Yossef, Jayram, Kumar, and Sivakumar~\cite{BJKS}, who 
  further improved this bound to $\Omega(n/k^{1+\gamma})$ for an arbitrarily 
  small constant $\gamma > 0$ in the one-way model of communication. These 
  bounds were further improved by Chakrabarti, Khot, and Sun~\cite{cks03} to 
  $\Omega(n/(k \log k))$ in the general communication model and an optimal 
  $\Omega(n/k)$ bound for $1$-way communication. The optimal $\Omega(n/k)$ 
  total communication bound for general communication was finally obtained in 
  \cite{g09,j09}. 

To illustrate a simple example of how this problem can be used, consider the {\it streaming model}. 
The streaming model is one of the most important models for processing massive datasets. 
One can model a stream as a list of integers $i_1, \ldots, i_m \in [n] = \{1, 2, \ldots, n\}$,
where each item $i \in [n]$ has a frequency $x_i$ which denotes its number of occurrences
in the stream. We refer the reader to \cite{BBDMW02,muthukrishnan2005data} 
for further background on the streaming model of computation. 

An important problem in this model
is computing the $p$-th frequency moment $F_p = \sum_{j=1}^n x_j^p$. To reduce from the promise set disjointness problem, the first player
runs a streaming algorithm on the items in its set, passes the state of the algorithm to the next
player, and so on. The total communication is $k \cdot s$, where $s$ is the amount of memory of the
streaming algorithm. Observe that in the first case of the promise we have $F_p \leq n$, while in 
the second case we have $F_p \geq k^p$. Setting $k = (2n)^{1/p}$ therefore implies an algorithm
estimating $F_p$ up to a factor better than $2$ can solve promise set disjointness and therefore
$k \cdot s = (2n)^{1/p} s = \Omega(n/(2n)^{1/p})$, that is, $s = \Omega(n^{1-2/p})$. For $p > 2$, this is known to
be best possible up to a constant factor \cite{BKSV14}. 

Notice that nothing substantial would change in this reduction if one were to change the second
case in the promise to instead say: (2) there is a unique element $j$ occurring in at least half of 
the sets, and the sets are otherwise disjoint. Indeed, in the above reduction, in one case
we have $F_p \geq (k/2)^p$, while in the second case we have $F_p \leq n$. This recovers the same
$\Omega(n^{1-2/p})$ lower bound, up to a constant factor. We call this new problem ``mostly'' set 
disjointness ($\disj$). 

While it is seemingly inconsequential to consider $\disj$ instead of promise set
disjointness, there are some peculiarities about this
problem that one cannot help but wonder about. 
In the promise set disjointness problem, there is a deterministic
protocol solving the problem with $O(n/k \log k + k)$ bits of communication -- we walk through
the players one at a time, and each indicates if its set size is smaller than $n/k$. Eventually we must
reach such a player, and when we do, that player posts its set to the blackboard. We then ask one other
player to confirm an intersection. Notice that there always must exist a player with a set of size
at most $n/k$ by the pigeonhole principle. On the other hand, for the $\disj$
problem, it does not seem so easy to achieve a deterministic protocol with $O(n/k \log k + k)$ bits
of communication. Indeed, in the worst case we could have up to $k/2$ players posting their entire set
to the blackboard, and still be unsure if we are in Case (1) or Case (2). 

More generally, is there
a gap in the dependence on the error probability of algorithms for promise set disjointness versus $\disj$? Even if one's main interest is in constant error probability protocols, is there anything that can be learned from this new problem?

\subsection{Our Results}\label{sec:results}
We more generally define $\disj$ so that in Case (2), there is an item
occurring in $l = \Theta(k)$ of the sets, though it is still convenient 
to think of $l = k/2$. Our main theorem is that $\disj$ requires $\Omega(n)$
communication to solve deterministically, or even with failure
probability $e^{-k}$. 

% \begin{theorem}
%   $\disj$ with $n$ elements, $k$ players, and $l = ck$ for constant
%   $c \in (0,1)$ requires $\Omega(n)$ bits of communication for failure
%   probability $e^{-C_c k}$, where $C_c$ is a constant depending on
%   $c$ with $C_{1/2} = 1$.
% \end{theorem}

\define{thm:main}{Theorem}{%
	$\disj$ with $n$ elements, $k$ players, and $l = ck$ for an absolute 
	constant $c \in (0,1)$ requires $\Omega(\min(n, n \frac{\log 
	(1/\delta)}{k}))$ bits of communication for failure probability $\delta$.
}
\state{thm:main}

This result does not have any restriction on the order of
communication, and is in the ``blackboard model'' where each message
is visible to all other players. We note that as $c \rightarrow 1$, our lower bound
goes to $0$, but for any absolute constant $c \in (0,1)$, we achieve the stated
$\Omega(\min(n, n \frac{\log (1/\delta)}{k}))$ lower bound. We did not explicitly compute
our lower bound as a function of $c$, as $c \rightarrow 1$. 

Notice that for constant $\delta$, Theorem \ref{thm:main} recovers the $\Omega(n/k)$ total
communication bound for promise set disjointness, which was the result of a long sequence of work.
Our proof of Theorem \ref{thm:main} gives a much simpler proof of an $\Omega(n/k)$ total
communication lower bound, avoiding Hellinger distance and Poincare inequalities altogether,
which were the main ingredients in obtaining the optimal $\Omega(n/k)$ lower bound for promise
set disjointness in previous work. Moreover, as far as we are aware, an $\Omega(n/k)$ lower
bound for the $\disj$ problem suffices to recover all of the lower bounds in applications that
promise set disjointness has been applied to. Unlike our work, however, existing lower bounds for promise set
disjointness do not give improved bounds for small error probability $\delta$. Indeed,
it is impossible for them to do so because of the deterministic protocol described above. We 
next use this bound in terms of $\delta$ to obtain the first lower bounds for deterministic
streaming algorithms and randomized $\delta$-error algorithms for a large number of problems. 

We note that other work on deterministic communication lower bounds for streaming, e.g., the work
of Chakrabarti and Kale \cite{CK16}, does not apply here. They study multi-party equality problems
and it is not clear how to use their fooling set arguments to prove a lower bound for $\disj$.
One of the challenges in designing a fooling set is the promise, namely, 
that a single item occurs on a constant fraction of the
players {\it and} all remaining items occur on at most one player. This promise is crucial for the
applications of $\disj$. 

%Our relaxation of the ``uniquely intersecting'' promise to $l < k$
%players having the element is necessary for this result.  For example,
%if $l = k$ one can solve the problem deterministically with
%$O((n/k)\log k)$ communication by having the first player with at most
%$2n/k$ elements send her entire set to the next player, who looks for an
%overlap.  Getting an $\Omega(n)$ communication bound is key for our
%implications for streaming algorithms.

We now formally introduce notation for the data stream model. In the streaming model, an integer vector $x$ 
is initialized to $0^n$ and undergoes a sequence of $L = \textrm{poly}(n)$ updates. The
streaming algorithm is typically allowed one (or a few) passes over the stream, and the goal
is to use a small amount of memory. We cannot afford to store the entire stream since $n$
and $L$ are typically very large. In 
this paper, we mostly restrict our focus to the {\it insertion-only model} where the updates 
to the vector are of the form $x \leftarrow x + \delta$ where $\delta\in \{e_1, 
\dotsc, e_n\}$ is a standard basis vector. There are also the turnstile data stream models 
in which
$x \leftarrow x + \delta$ where $\delta \in \{e_1, \dotsc, e_n, -e_1, \dotsc, -e_n\}$. In the
{\it strict turnstile model} it is promised that $x \geq 0^n$ at all times in the stream,
whereas in the {\it general turnstile model} there are no restrictions on $x$. Therefore,
an algorithm in the general turnstile model works also in the strict turnstile model
and insertion-only models. 

\paragraph{Finding Heavy Hitters.} 
Finding the heavy hitters, or frequent items, 
is one of the most fundamental problems in data streams. 
These are useful in IP routers \cite{ev03}, in association rules and frequent itemsets \cite{as94,son95,toi96,hid99,hpy00}
and databases \cite{fsgmu98,br99,hpdw01}. Finding the heavy hitters is also frequently used as a subroutine
in data stream algorithms for other problems, such as moment estimation \cite{IndykW05}, 
entropy estimation \cite{ChakrabartiCM10,HarveyNO08}, $\ell_p$-sampling \cite{MonemizadehW10}, finding duplicates \cite{GopalanR09}, and so on.
For surveys on algorithms for heavy hitters, 
see, e.g., \cite{CormodeH08,woodruff2016new}. 

In the {\it $\epsilon$-$\ell_p$-heavy hitters problem}, for $p \geq 1$, 
the goal is to find a set $S$ which contains all 
indices $i\in [n]$ for which $\abs{x_i}^p \geq \epsilon^p \norm{p}{x}^p$, and
contains no indices $i \in [n]$ for which $\abs{x_i}^p \leq 
\frac{\epsilon^p}{2} \norm{p}{x}^p$.

The first heavy hitters algorithms were for $p = 1$, given 
by Misra and Gries \cite{MG82}, who achieved $O(\epsilon^{-1})$
words of memory, where a word consists of $O(\log n)$ bits of space. 
Interestingly, their algorithm is {\it deterministic}, i.e., the failure
probability $\delta = 0$. 
This algorithm was rediscovered
by Demaine, L\'opez-Ortiz, and Munro~\cite{demaine2002frequency}, and again by Karp, Shenker, and Papadimitriou~\cite{karp2003simple}. 
Other than these 
algorithms, which are deterministic, there are a number of randomized
algorithms, such as the Count-Min sketch \cite{cormode2005improved}, sticky sampling \cite{mm02},
lossy counting \cite{mm02}, space-saving \cite{MetwallyAA05}, sample and hold \cite{ev03}, multi-stage bloom filters
\cite{cfm09}, and sketch-guided sampling \cite{kx06}. One can also achieve stronger residual error guarantees~\cite{bics10}.

An often much stronger notion than an $\ell_1$-heavy hitter is an $\ell_2$-heavy hitter. 
Consider an $n$-dimensional vector $x = (\sqrt{n}, 1, 1, \ldots, 1)$. The 
first coordinate is an $\ell_2$-heavy hitter with parameter $\epsilon = 
1/\sqrt{2}$, but it is only an $\ell_1$
heavy hitter with parameter $\epsilon = 1/\sqrt{n}$. Thus, the algorithms above would
require at least $\sqrt{n}$ words of memory to find this heavy hitter. 
In \cite{ccf04} this problem was solved by the {\sc CountSketch} algorithm, which provides
a solution to the $\epsilon$-$\ell_2$-heavy hitters problem, and more generally to the 
$\ell_p$-heavy
hitters problem for any $p \in (0,2]$\footnote{For $p < 1$ the 
quantity $\norm{p}{x}$ is not a norm, but it is still a well-defined quantity.}, in $1$-pass
and in the general turnstile model using 
$O\left (\frac{1}{\epsilon^p}\log(n/\delta) \right )$ words of memory. 
For insertion-only streams, this
was recently improved \cite{bciw15,bcinww16} 
to $O \left (\frac{1}{\epsilon^p} \right )$ words of memory for constant
$\delta$, and $O \left (\frac{1}{\epsilon^p} \log(1/\delta) \right )$ in 
general. See also work \cite{LNNT19} on reducing the decoding time for finding the heavy
hitters from the algorithm's memory contents, without sacrificing additional
memory. 

There is also work establishing lower bounds for heavy hitters.  The
works of \cite{bipw11,JST11} establish an
$\Omega \left (\frac{1}{\epsilon^p} \log n \right )$ word lower bound
for any value of $p > 0$ and constant $\delta$, for any algorithm in
the strict turnstile model. This shows that the above algorithms are
optimal for constant $\delta$. Also for $p > 2$, it is known that
solving the $\epsilon$-$\ell_p$-heavy hitters problem even with
constant $\epsilon$ and $\delta$ requires $\Omega(n^{1-2/p})$ words of
memory \cite{BJKS,g09,j09}, and thus $p = 2$ is often considered the
gold standard for space-efficient streaming algorithms since it is the
largest value of $p$ for which there is a poly$(\log n)$ space algorithm.  For
deterministic algorithms computing linear sketches, the work
of~\cite{G092} shows the sketch requires
$\Omega(n^{2-2/p}/\epsilon^2)$ dimensions for $p \geq 1$ (also shown
for $p=2$ by~\cite{CDD09}). This also implies a lower bound for
general turnstile algorithms for streams with several important
restrictions; see also \cite{li2014turnstile,k20}. There is also work
on the related compressed sensing problem which studies small $\delta$
\cite{g13}.

Despite the work above, for all we knew it could be entirely possible that, in the 
insertion-only model, an 
$\epsilon$-$\ell_2$-heavy hitters algorithm could achieve 
$O \left (\frac{1}{\epsilon^2} \right )$ words of memory and 
solve the problem {\it deterministically}, i.e., with $\delta = 0$. In fact, it is well-known
that the above $\Omega(n)$ lower bound for $\epsilon$-$\ell_2$-heavy hitters for linear
sketches does not hold in the insertion-only model. Indeed, by 
running a deterministic algorithm for $\epsilon$-$\ell_1$-heavy hitters,
we have that if $x_i^2 \geq \epsilon^2 \|x\|_2^2$, then $x_i \geq \epsilon \|x\|_2 \geq \frac{\epsilon}{\sqrt{n}} \|x\|_1$, and consequently one can find all $\ell_2$-heavy hitters using
$O \left (\frac{\sqrt{n}}{\epsilon} \right )$ words of memory. Thus, for constant $\epsilon$,
there is a deterministic  $O \left (\sqrt{n} \right )$ words of memory upper bound, but
only a trivial $\Omega \left (1 \right )$ word lower bound. Surprisingly, this factor $\sqrt{n}$ gap was left wide open, and the main question we ask about heavy hitters is:
\begin{center}
{\it Can one deterministically solve $\epsilon$-$\ell_2$-heavy hitters in insertion-only streams in constant memory?}
\end{center}

One approach to solve $\disj$ would be for each player to insert their
elements into a stream and apply a heavy hitters algorithm.  For example, if
$k = \sqrt{n}$, there will be a $\Theta(1)$-$\ell_2$-heavy hitter if
and only if the $\disj$ instance is a YES instance.  For a space-$S$
streaming algorithm, this uses $Sk$ communication to pass the
structure from player to player.  Hence $S \gtrsim n/k = \sqrt{n}$.  In
general:

\define{thm:delta-lb}{Theorem}{%
  Given $\eps\in (\frac{1}{n^{1/p}}, \frac{1}{2})$ and $p \geq 1$, any
  $\delta$-error $r$-pass insertion-only streaming algorithm for
  $\eps$-$\ell_p$-heavy hitters requires 
  $\Omega(\min(\frac{n^{1-1/p}}{r\eps}, n^{1 - 2/p}\frac{\log(1/\delta)}{r \eps^2}))$ bits of space. 
}
\restate{thm:delta-lb}

Most notably, setting $\delta = 0$ and $p = 2$ and $r = O(1)$, this
gives an $\Omega(\sqrt{n}/\eps)$ bound for deterministic $\ell_2$
heavy hitters.  The \textsc{FrequentElements} algorithm~\cite{MG82}
matches this up to a factor of $\log n$ (i.e., it uses this many
words, not bits).
For $n^{-.1} > \delta > 0$, the other term
($\frac{\log(1/\delta)}{r \eps^2}$) is also achievable up to the
bit/word distinction, this time by~\textsc{CountSketch}.  For larger
$\delta$, we note that it takes 
$\Omega(\frac{1}{\eps^2}\log \eps^2 n)$ bits already to encode the output
size. As a result, we show that the existing algorithms are within a $\log n$ factor of optimal.

One common motivation for heavy hitters is that many distributions are
power-law or Zipfian distributions.  For such distributions, the $i$-th
most frequent element has frequency approximately proportional to
$i^{-\zeta}$ for some constant $\zeta$, typically
$\zeta \in (0.5, 1)$~\cite{clauset2009power}.  Such distributions have
significant $\ell_{1/\zeta}$-heavy hitters.  Despite our lower bound
for general heavy hitters, one might hope for more efficient
deterministic/very high probability insertion-only algorithms in this
special case.  We rule this out as well, getting an
$\Omega(\min(n^{1-\zeta}, n^{1-2\zeta}\log(1/\delta)))$ lower bound
for finding the heavy hitters of these distributions (see Theorem~\ref{thm:power-law}).  This again
matches the upper bounds from \textsc{FrequentElements} or
\textsc{CountSketch} up to a logarithmic factor.

To extend our lower bound to power-law distributions, we embed our hard
instance as the single largest and $n/2$ smallest entries of a
power-law distribution; we then insert the rest of the power-law
distribution deterministically, so the overall distribution is
power-law distributed.  Solving heavy hitters will identify whether
this single largest element exists or not, solving the communication
problem.

\define{thm:p-pass-lb}{Theorem}{%
  For any $\eps\in (\frac{1}{n^{1/p}}, \frac{1}{2})$ and $p \geq 1$,
  any $r$-pass deterministic insertion-only streaming algorithm for
  $\eps$-$\ell_p$-heavy hitters must have a space complexity of
  $\Omega(\frac{n^{1-1/p}}{r\eps})$ bits.
}

\paragraph{Frequency Moments.}
We next turn to the problem of estimating the frequency moments $F_p$, 
which in our reduction from
the $\disj$ problem, corresponds to estimating $\|x\|_p^p = \sum_{i=1}^n 
|x_i|^p$. 
Our hard instance for $\disj$ immediately gives us the following theorem:

\begin{theorem}\label{thm:delta-lb-fp}{
  For any constant $\epsilon \in (0,1)$ and $p \geq 2$, any
  $\delta$-error $r$-pass insertion-only streaming algorithm for
  $\eps$-$F_p$-estimation must have space complexity of
  $\Omega(\min(\frac{n^{1-1/p}}{r}, n^{1 - 2/p}\frac{\log(1/\delta)}{r}))$ bits.
}
\end{theorem}

The proof of Theorem \ref{thm:delta-lb-fp} follows immediately by setting the number of 
players in $\disj$ to be $\Theta((\epsilon n)^{1/p})$, and performing the reduction to
$F_p$-estimation described before Section \ref{sec:results}. This improves the 
previous $\Omega((n^{1-2/p} + \log(1/\delta))/r)$ lower bound, which follows from
\cite{BJKS,j09}, as well as a simple reduction from the Equality function \cite{AMS99}, see also \cite{CK16}. 
It matches an upper bound of \cite{BKSV14} for constant $\epsilon$, by repeating their algorithm independently $O(\log(1/\delta))$ times. 
Our lower bound instance shows that to approximate 
$\|x\|_{\infty} = \max_i |x_i|$ of an integer vector, with
$O(\log n)$-bit coordinates in $n$ dimensions, up to an additive
$\Theta(\sqrt{\|x\|_2})$ deterministically, one needs
$\Omega(\sqrt{n})$ memory. This follows from our hard instance. Approximating the $\ell_{\infty}$ norm is an
important problem in streaming, and its complexity was asked about in
Question 3 of \cite{c05}.

%\paragraph{Other implications for streaming.}
\paragraph{Low Rank Approximation.}
Our $\ell_2$-heavy hitters lower 
bound also has applications to {\it deterministic} low rank 
approximation in a stream, a topic of recent interest
\cite{l13,gp14,w14,g16,glp16,h18}. 
Here we see rows 
$A_1, A_2, \ldots, A_n$ of an $n \times d$ matrix $A$ one at a time. 
At the end of 
the stream we should output a projection $P$ onto a rank-$k$ space for which 
$\|A-AP\|_F^2 \leq (1+\epsilon)\|A-A_k\|_F^2$, where $A_k$ is the best rank-$k$ 
approximation to $A$.
 A natural question is if the deterministic 
{\sf FrequentDirections} algorithm of \cite{g16} using $O(dk/\epsilon)$
words of memory 
can be improved when the rows 
of $A$ are $O(1)$-sparse. 
The sparse setting was shown to have faster 
running times in \cite{glp16,h18}, and more efficient randomized communication
protocols in \cite{bwz16}. Via a reduction from our $\disj$ problem, we show a polynomial dependence on $d$ is necessary.

\define{thm:rank}{Theorem}{%
Any $1$-pass deterministic streaming algorithm outputting a rank-$k$ projection 
matrix $P$ providing a $(1+\epsilon)$-approximate rank-$k$ low rank approximation requires $\Omega(\sqrt{d})$ bits of memory, even for $k = 1$, 
$\epsilon = \Theta(1)$, and when each row of $A$ has only a single non-zero entry.
}
\restate{thm:rank}

\paragraph{Algorithms and Lower Bounds in Other Streaming Models.}
We saw above that deterministic insertion-only $\ell_2$ heavy hitters
requires $\Thetat(\sqrt{n})$ space for constant $\eps$.  We now
consider turnstile streaming and linear sketching.

The work of~\cite{G092,CDD09} shows that $\Omega(n)$ space is needed
for general deterministic linear sketching, but the corresponding hard
instances have negative entries.  We extend this in two ways: when
negative entries are allowed, an $\Omega(n)$ lower bound is easy even
in turnstile streaming (for heavy hitters, but not the closely related
$\ell_\infty/\ell_2$ sparse recovery guarantee; see Remark~\ref{remark:turn}).  If
negative entries are not allowed, we still get an $\Omega(n)$ bound on
the number of linear measurements for deterministic linear sketching
(see Theorem~\ref{thm:linear-sketching}).

A question is if we can solve $\ell_2$ heavy hitters
deterministically in the strict turnstile model in $o(n)$ space.  In
some sense the answer is no, due to the near equivalence between
turnstile streaming and linear sketching~\cite{G092,LNW14,AHLW16}, but
this equivalence has significant limitations.  Recent work has shown
that with relatively mild restrictions on the stream, such as a bound
on the length $L$, significant improvements over linear sketching are
possible~\cite{jayaram2018data,k20}.  Can we get that here?
We show that this is indeed possible: streams with $O(n)$ updates can
be solved in $O(n^{2/3})$ space.  While this does not reach the
$\sqrt{n}$ lower bound from insertion-only streams
(Theorem~\ref{thm:p-pass-lb}), it is significantly better
than the $\Omega(n)$ for linear sketches.  In general, we show:

\define{thm:alg}{Theorem}{%
  There is a deterministic $\ell_2$ heavy hitters algorithm for
  length-$L$ strict turnstile streams with $\pm 1$ updates using
  $O((L/\eps)^{2/3})$ words of space.
}
\restate{thm:alg}

Our algorithm for short strict turnstile streams is a combination of
\textsc{FrequentElements} and exact sparse recovery.  With space $S$,
\textsc{FrequentElements} (modified to handle negative updates) gives
estimation error $L/S$, which is good unless $\|x\|_2 \ll L/S$.  But
if it is not good, then $\|x\|_0 \leq \|x\|_2^2 \ll (L/S)^2$.  Hence in that
case $(L/S)^2$-sparse recovery will recover the vector (and hence the
heavy hitters).  Running both algorithms and combining the results
takes $S + (L/S)^2$ space, which is optimized at $L^{2/3}$.

\subsection{Our Techniques}
Our key lemma is that solving $\disj$ on $n$ elements, $k$ items, and
$l = ck$ with probability $1 - e^{-k}$ has $\Omega(n)$ conditional information
complexity for any constant $c \in (0, 1)$.  It is well-known that the
conditional information complexity of a problem lower bounds its 
communication complexity (see, e.g., \cite{BJKS}).

This can then be extended
to $\delta \gg e^{-\Theta(k)}$ using repetition, namely, we can amplify
the success probability of the protocol to $1-e^{-\Theta(k)}$ by independent 
repetition, 
apply our $\Omega(n)$ 
lower bound on the new protocol with $\delta = e^{-\Theta(k)}$, 
and then conclude a lower bound on the original protocol. Indeed, this is how
we obtain our total communication lower bound of $\Omega(n/k)$ for constant
$\delta$, providing a much simpler proof than that of the $\Omega(n/k)$ total
communication lower bound for promise set disjointness in prior work.

Our bound is tight up to a $\log k$ factor.  It can be solved
deterministically with $O(n \log k)$ communication (for each bit, the
first player with that bit publishes it), and with probability
$1-(1-\eps)^{l-1}$ using $O(\eps n \log k)$ communication (only
publish the bit with probability $\eps$).  Setting $\eps = o(1)$, any
$e^{-o(k)}$ failure probability is possible with $o(n \log k)$ communication.

%{\bf Lower bound for $\disj$.}  
We lower bound $\disj$ using conditional 
information complexity.  Using the direct sum property of conditional 
information cost, analogous to previous work (see, e.g., \cite{BJKS}), it 
suffices to get an $\Omega(1)$ conditional information cost bound for the $n=1$ 
problem $F_k$: we have $k$ players, each of whom receives one bit, and the 
players must distinguish (with probability $1 - e^{-k}$) between at most one 
player having a $1$, and at least $\Omega(k)$ players having $1$s.  In
particular, it suffices to show for correct protocols $\pi$ that
\begin{align}\label{eq:dtv}
  \E_{i \in [t]} \dtv(\pi_0, \pi_{e_i}) = \Omega(1)
\end{align}
where $\pi_0$ is the distribution of protocol transcripts if the
players all receive $0$, and $\pi_{e_i}$ is the distribution if player
$i$ receives a $1$. The main challenge is therefore in bounding this
expression. 

Consider any protocol that does not satisfy~\eqref{eq:dtv}.  We show
that, when $\dtv(\pi_0, \pi_{e_i}) \ll 1$, player $i$ can be
implemented with an equivalent protocol for which 
the player usually does not even observe its input bit.  That is, if
every other player receives a $0$, player $i$ will only observe its bit
with probability $\dtv(\pi_0, \pi_{e_i})$.  This means that most
players only have a small probability of observing their bit.  The
probability that any two players $i, i'$ observe their bits may be
correlated; still, we show that this implies the existence of a large
set $S$ of $ck$ players such that the probability---if every player
receives a zero---that \emph{no} player $i \in S$ observes their bit
throughout the protocol is above $e^{-k}$.  But then
$\dtv(\pi_0, \pi_{e_S}) < 1 - e^{-k}$, so the protocol cannot
distinguish these cases with the desired probability. We now give the
full proof.

\section {Preliminaries}

We use the following measures of distance between distributions in our proofs. 
\begin{definition}
	Let $P$ and $Q$ be probability distributions over the same countable 
	universe $\mathcal{U}$. The total variation distance between $P$ and $Q$ 
	is defined as:
	%\[
		$\dtv(P, Q) = \frac{1}{2} \norm{1}{P - Q}.$
	%\]
%	The squared Hellinger distance between $P$ and $Q$ is defined as:
%	\[
%		h^2(P, Q) = 1 - \sum_{x\in \mathcal{U}} \sqrt{P(x)\cdot Q(x)} = 
%		\frac{1}{2} \cdot \sum_{x\in \mathcal{U}} \big(\sqrt{P(x)} - 
%		\sqrt{Q(x)} \big).
%	\]
\end{definition}

%In this paper, we will sometimes abuse notation and consider distances between 
%random variables instead of the underlying distributions. These two notions of 
%distance between distributions are closely related:
%\begin{lemma}\label{lemma:dtv-hellinger}
%	For any two probability distributions $P$ and $Q$, it holds that:
%	\[
%		h^2(P, Q) \leq \dtv(P, Q) \leq h(P, Q) \cdot \sqrt{2 - h^2(P, 
%		Q)}. 
%	\]
%\end{lemma}

In our proof we also use the Jensen-Shannon divergence and Kullback-Liebler 
divergence. We define these notions of divergence here:
\begin{definition}
	Let $P$ and $Q$ be probability distributions over the same discrete
	universe $\mathcal{U}$. The Kullback-Liebler divergence or KL-divergence 
	from $Q$ to $P$ is defined as:
	%\[
	$D_{KL}(P, Q) = \sum_{x\in \mathcal{U}} P(x)\log(\frac{P(x)}{Q(x)}).$
	%\]
	This is an asymmetric notion of divergence. The Jensen-Shannon 
	divergence between two distributions $P$ and $Q$ is the symmetrized 
	version of the KL divergence, defined as: 
	%\[
	$D_{JS}(P, Q) = \frac{1}{2}(D_{KL}(P, Q) + D_{KL}(Q, P).$
	%\]
\end{definition}
From Pinsker's inequality, for any two distributions $P$ and $Q$, 
$D_{KL}(P, Q) \geq \frac{1}{2} \dtv^2(P, Q)$.

%\subsection{Multiparty Communication}
In the multiparty communication model we consider $k$-ary 
functions $F: \mathcal{L} \rightarrow\mathcal{Z}$ where $\mathcal{L} \subseteq 
\mathcal{X}_1 \times \mathcal{X}_2 \times \cdots \times \mathcal{X}_k$. There 
are $k$ parties(or players) who receive inputs $X_1, \dotsc, X_k$ which are 
jointly distributed according to some distribution $\mu$. We consider protocols 
in the blackboard model where in any protocol $\pi$ players speak in any order 
and 
each player broadcasts their message to all other players. 
So, the message of player $i$ is a function of the messages they receive, their 
input and randomness i.e., $m_i = M_i(X_i, m_{i-1}, R_i)$. The final player's 
message is the output of the protocol.

The communication cost of a multiparty protocol $\pi$ is the sum of the lengths 
of the individual messages $\norm{}{\pi} = \sum \abs{M_j}$.
A protocol $\pi$ is a $\delta$-error protocol for the function $f$ if for every 
input $x\in \mathcal{L}$, the output of the protocol equals $f(x)$ with 
probability 
$1-\delta$. 
The randomized communication complexity of $f$, denoted  
$R_{\delta}(f)$, is the cost of the cheapest randomized protocol that computes 
$f$ correctly on every input with error at most $\delta$ over the randomness of 
the protocol. 

The distributional communication complexity of the function $f$ for error 
parameter $\delta$ is denoted as $D_{\mu}^{\delta}(f)$. This is the 
communication cost of the cheapest deterministic protocol which computes the 
function $f$ with error at most $\delta$ under the input distribution $\mu$.
By Yao's minimax theorem, $R_{\delta}(f) = \max_{\mu} D_{\mu}^{\delta}(f)$ and 
hence it suffices to prove a lower bound for a hard distribution $\mu$.
In our proofs, we bound the conditional information complexity of a function in 
order to prove lower bounds on $R_\delta(f)$. We define this notion below.
\begin{definition}
	Let $\pi$ be a randomized protocol whose inputs belong to 
	$\mathcal{K}\subseteq \mathcal{X}_1 \times \mathcal{X}_2 \dotsc \times 
	\mathcal{X}_k$. Suppose $((X_1, X_2, \dotsc, X_k), D)\sim \eta$ where 
	$\eta$ is a distribution over $\mathcal{K}\times\mathcal{D}$ for some 
	set $\mathcal{D}$. The \textbf{conditional information cost} of $\pi$ 
	with respect to $\eta$ is defined as:
	%\[
	$cCost_{\eta}(\pi) = I(X_1, \dotsc, X_k; \pi(X_1, \dotsc, X_k) \mid D).$
	%\] 
\end{definition}
\begin{definition}
	The $\delta$-error \textbf{conditional information complexity} of $f$ 
	with respect to $\eta$, denoted $CIC_{\eta, \delta}(f)$ is defined as 
	the 
	minimum conditional information cost of a $\delta$-error protocol for 
	$f$ with respect to $\eta$.
\end{definition}
In \cite{BJKS} it was shown that the randomized communication complexity 
of a function is at least the conditional information complexity of the 
function $f$ with respect to any input distribution $\eta$. 
\begin{propn}[Corollary 4.7 of \cite{BJKS}]\label{cor:rand-cic}
	Let $f: \mathcal{K}\rightarrow \{0,1\}$, and let $\eta$ be a 
	distribution over $\mathcal{K}\times \mathcal{D}$ for some set 
	$\mathcal{D}$. Then, $R_\delta(f)\geq CIC_{\eta, \delta}(f)$.
\end{propn}

\paragraph{Direct Sum.}
Per~\cite{BJKS}, conditional information complexity obeys a Direct Sum
Theorem condition under various conditions.
The Direct Sum Theorem of~\cite{BJKS} allows us to reduce a $t$-player
conditional information complexity problem with an $n$-dimensional input
to each player to a $t$-player conditional information complexity with
a $1$-dimensional input to each player. This theorem applies when the
function is ``decomposable'' and the input distribution is
``collapsing''. We define both these notions here.
\begin{definition} 
	Suppose $\mathcal{L} \subseteq \mathcal{X}_1\times \mathcal{X}_2 \times 
	\dotsc \times \mathcal{X}_t$ and $\mathcal{L}_n\subseteq \mathcal{L}^n$. 
	A function $f: \mathcal{L}_n \rightarrow \{0,1\}$ 
	is $g$\textbf{-decomposable} with primitive $h:\mathcal{L}\rightarrow 
	\{0,1\}$ if it can be written as:
	\[
		f(X_1, \dotsc, X_t) = g(h(X_{1,1}, \dotsc, X_{1,t}), \dotsc, 
		h(X_{n,1}, \dotsc, X_{n,t}))
	\]
	for $g:\{0,1\}^n\rightarrow \{0,1\}$.
\end{definition}

\begin{definition}
	Suppose $\mathcal{L} \subseteq \mathcal{X}_1\times \mathcal{X}_2 \times 
	\dotsc \times \mathcal{X}_t$ and $\mathcal{L}_n\subseteq \mathcal{L}^n$. 
	A distribution $\eta$ over $\mathcal{L}_n$ is a \textbf{collapsing  
	distribution} for $f: \mathcal{L}_n \rightarrow \{0,1\}$ with 	 
	respect to $h:\mathcal{L}\rightarrow \{0,1\}$ if for all $Y_1, \dotsc, 
	Y_n$ in the support of $\eta$, for all $y\in \mathcal{L}$ and for all 
	$i\in [n]$, $f(Y_1, \dotsc, Y_{i-1}, y, Y_{i+1}, \dotsc, Y_n) = h(y).$
\end{definition}
We state the Direct Sum Theorem for conditional information complexity below. 
The proof of this theorem in \cite{BJKS} applies to the blackboard model of 
multiparty communication. We state this in the most general form here and then 
show that it may be applied to the hard distribution $\eta_0$ which we choose in 
Section \ref{sec:comm-lb}.
\begin{theorem}[Multiparty version of Theorem 5.6 of 
\cite{BJKS}]\label{thm:dirsum-formal}
	Let $\mathcal{L} \subseteq \mathcal{X}_1\times \mathcal{X}_2 \times 
	\dotsc \mathcal{X}_t$ and let $\mathcal{L}_n \subseteq \mathcal{L}^n$. 
	Suppose that the following conditions hold:
	\begin{itemize}
		\item[(i)] $f: \mathcal{L}_n \rightarrow \{0,1\}$ is a 
		decomposable function with primitive $h:\mathcal{L}\rightarrow 
		\{0,1\}$,
		\item[(ii)] $\zeta$ is a distribution over $\mathcal{L}\times 
		\mathcal{D}$, such that for any $d\in \mathcal{D}$ 
		the distribution $(\zeta \mid D=d)$ is a product distribution,
		\item[(iii)] $\eta = \zeta^n$ is supported on 
		$\mathcal{L}_n\times \mathcal{D}^n$, and
		\item[(iv)] the marginal probability distribution of $\eta$ over 
		$\mathcal{L}_n$ is a collapsing distribution for $f$ with respect 
		to $h$.
	\end{itemize}
	Then $CIC_{\eta, \delta}(f) \geq n\cdot CIC_{\zeta, \delta}(h).$
\end{theorem}

\section{Communication Lower Bound for Mostly Set Disjointness}\label{sec:comm-lb}
Let $[n] = \{1, 2, \ldots, n\}$. We let $H(X)$ denote the entropy of a random variable
$X$, and $I(X ; Y) = H(X) - H(X | Y)$ be the mutual information. 

\subsection{The Hard Distribution}
%In this section, we prove lower bounds on the deterministic and randomized 
%multiparty communication complexity of the Mostly Set Disjointness problem. 

%\subsection{Mostly Set Disjointness}\label{subsec:mostlydisj}
\begin{definition}
	Denote by $\disj_{n, l, t}$, the multiparty Mostly 
	Set-Disjointness problem in which each player $j\in[t]$ receives an 
	$n$-dimensional input vector $X_j = (X_{j,1}, \dotsc, X_{j,n})$ where 
	$X_{j,i}\in \{ 0,1\}$ and the input to the protocol falls into either of 
	the following cases:
	\begin{itemize}
		\item \textbf{NO:} For all $i\in[n]$, $\sum_{j\in [t]} 
		X_{j,i}\leq 1$
		\item \textbf{YES:} There exists a unique $i\in[n]$ such that 
		$\sum_{j\in [t]} X_{j,i} = l$ and for all other $i'\neq i,  
		\sum_{j\in [t]} X_{j,i'}\leq 1$.
	\end{itemize}
	%At the end of the protocol. 
        The final player must output $1$ 
	if the input is in the YES case and $0$ in the NO case. 
\end{definition}
Let $\mathcal{L}\subset \{0, 1\}^t$ be the set of valid inputs along one index 
in $[n]$ for $\disj_{n,l,t}$, i.e., the set of elements in $x \in 
\{0, 1\}^t$ with $\sum_{j\in [t]} x_j \leq 1$ or $\sum_{j\in [t]} x_j = l$. 
Let $\mathcal{L}_n\subset \mathcal{L}^n$ denote the set of valid inputs to the 
$\disj_{n,l,t}$ function.

We define $\disj_{n,l,t}: \mathcal{L}_n \rightarrow \{0,1\}$ 
as:
%\[
$\disj_{n, l, t}(X_1, \dotsc, X_t) = \bigvee_{i\in [n]} F_{l, t}(X_{1,i}, 
\dotsc, X_{t, i})$
%\]
where $F_{l,t}:\mathcal{L}\rightarrow \{0,1\}$ defined as:
%\[
$F_{l,t}(x_1, \dotsc, x_t) = \bigvee_{\substack{S\subseteq [t]\\\abs{S} = 
		l}} \bigwedge_{j\in S} x_j.$
%\]
This means $\disj_{n, l, t}$ is \textsf{OR}-decomposable into $n$ copies 
of $F_{l,t}$ and we may hope to apply a direct sum theorem with an appropriate 
distribution over the inputs.

In order to prove a lower bound on the conditional information complexity, we 
need to define a ``hard'' distribution over the inputs to $\disj_{n,l,t}$. We 
define the distribution $\eta$ over $\mathcal{L}_n\times \mathcal{D}^n$ where 
$\mathcal{D} = [t]$ as follows:
\begin{itemize}
	\item For each $i\in [n]$ pick $D_i\in[t]$ uniformly at random and sample
	$X_{D_i,i}$ uniformly from $\{0 ,1\}$  and for all $j'\neq D_i$ set 
	$X_{j', i}=0$. 
	\item Pick $I\in [n]$ uniformly at random and $Z\in \{0,1\}$
	\item if $Z=1$, pick a set $S\subseteq [t]$ such that $\abs{S}=l$ 
	uniformly at random and for all $j\in S$ set $X_{j,I}=1$ and for all 
	$j\notin S$, set $X_{j,I}=0$
\end{itemize} 

Let $\mu_0$ denote the distribution for each $i\in [n]$ conditioned on $Z = 0$. 
For any $d\in [t]$, when $D = d$, the conditional distribution over 
$\mathcal{L}$ is the uniform distribution over $\{ 0, e_d\}$ and hence a 
product distribution. Let $\eta_0$ be the distribution $\eta$ conditioned on 
$Z=0$. Clearly, $\eta_0 = \mu_0^n$.

This definition of $\disj_{n,l,t}$ and the hard distribution $\eta_0$
allows us to apply the Direct Sum theorem (Theorem
\ref{thm:dirsum-formal}) of \cite{BJKS}.  Note that: (i)
$\disj_{n,l,t}$ is OR-decomposable by $F_{l,t}$, (ii) $\mu_0$ is a
distribution over $\mathcal{L} \times [t]$ such that the marginal
distribution $(\mu_0 \mid D = d)$ over $\mathcal{L}$ is uniform
over $\{0, e_d\}$(and hence a product distribution), (iii)
$\eta_0 = \mu_0^n$, and (iv) since $\disj_{n,l,t}$ is OR-decomposable
and $\eta_0$ has support only on inputs in the NO case, $\eta_0$ is a
collapsing distribution for $\disj_{n,l,t}$ with respect to $F_{l,t}$.  Hence:
\begin{align}\label{eq:disj-dirsum}
    CIC_{\eta_0, \delta}(\disj_{n,l,t}) \geq n\cdot CIC_{\mu_0, 
      \delta}(F_{l,t})
\end{align}

\subsection{Information Cost for a Single Bit}
A key lemma for our argument is that the players can be implemented so
that they only ``observe'' their input bits with small probability.
The model here is that each player's input starts out hidden, but they
can at any time choose to observe their input.  Before they observe
their input, however, all their decisions (including messages sent and
choice of whether to observe) depend on the transcript and randomness,
but not the player's input.

In this section we use $\pi$ to denote the protocol in consideration and abuse 
notation slightly by using $\pi_x$ to denote the distribution of the transcript 
of the protocol $\pi$ on input $x$. % We clarify this difference wherever it 
%may cause confusion.

\begin{definition}
  Any (possibly multi-round) communication protocol involving $n$
  players, where each player receives one input bit, is defined to be
  a ``clean'' protocol with respect to player $i$ if, in each round,
  \begin{enumerate}
  \item if player $i$ has previously not ``observed'' his input bit,
    he ``observes'' his input bit with some probability that is a
    function only of the previous messages in the protocol,
  \item if player $i$ has not observed his input bit in this round or
    any previous round, then his message distribution depends only on
    the previous messages in the protocol but not his input bit, and
  \item if player $i$ \emph{has} observed his input bit in this round
    or any previous round, then---for a fixed value of the previous
    messages in the protocol---his distribution of messages on input
    $0$ and on input $1$ are \emph{disjoint}.
  \end{enumerate}    
\end{definition}

\begin{figure}[h!]
\centering
\includegraphics[width=\textwidth]{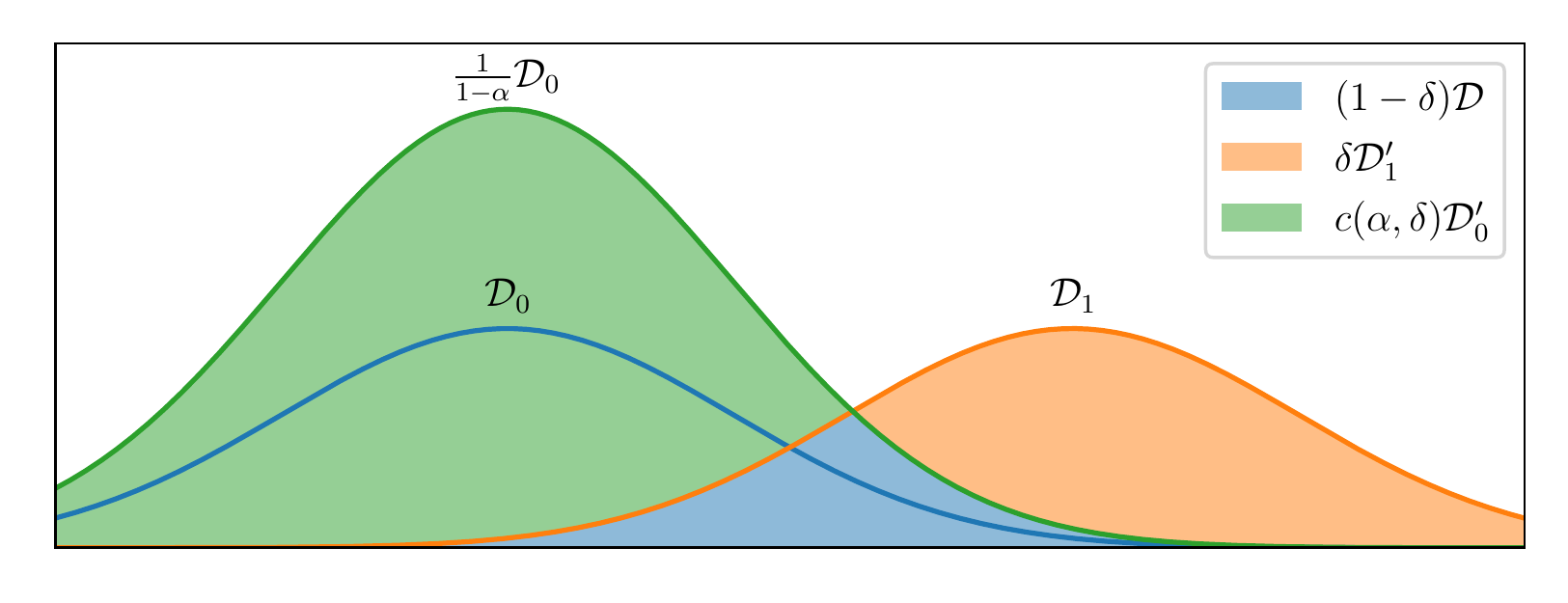}
\caption{An illustration of Lemma~\ref{lem:distributions}, given a parameter $\alpha$ and pair of distributions $(\cD_0,\cD_1)$.  We set $(1-\delta) \cD$ to be the overlap between $\cD_1$ and $\frac{1}{1-\alpha}\cD_0$, then $\cD_0'$ and $\cD_1'$ to be proportional to the remainder of $\frac{1}{1-\alpha}\cD_0$ and $\cD_1$, respectively.  These $\cD_0'$ and $\cD_1'$ are disjoint.}
\label{fig:dist}
\end{figure}

We start off by proving a lemma about decomposing any two arbitrary
distributions into one ``common'' distribution and two disjoint
different distributions. This lemma will enable us to show that any
communication protocol can be simulated in a clean manner.  

\begin{lemma}\label{lem:distributions}
	Let $\cD_0, \cD_1$ be two distributions, and $\alpha \in [0, 1]$. There
	exist three distributions $\cD, \cD_0', \cD_1'$ and a parameter
	$\delta \in (0, 1)$ such that:
	%\begin{align*}
	$\cD_0 = (1-\alpha)(1-\delta)\cD +  (1 - (1-\alpha)(1-\delta))\cD_0', 
	\cD_1 = (1-\delta)\cD + \delta \cD_1',$
	%\end{align*}
	and $\cD_0'$ has a disjoint support from $\cD_1'$.
\end{lemma}
We refer the reader to Figure \ref{fig:dist} for an illustration corresponding to Lemma
\ref{lem:distributions}. 
\begin{proof}
  We begin with two special cases.  If $\alpha=1$, then setting
  $\delta = 0$ allows us to set $\cD_0'=\cD_0$, $\cD =
  \cD_1$. $\cD_1'$ may be any arbitrary distribution that has disjoint
  support from $\cD_0'$.  If
  $\supp(\cD_0)\cap \supp(\cD_1) = \emptyset$, we may set
  $\delta = 1$, $\cD_0' = \cD_0$ and $\cD_1'=\cD_1$.
	
	So it suffices to consider the case where $\alpha < 1$ and  
	$\supp(\cD_0)\cap \supp(\cD_1) \neq \emptyset$. Let $\cD$ and $\delta$ 
	be such that $\cD(x) = \frac{1}{1-\delta}  
	\min(\frac{1}{1-\alpha}\cD_0(x),\cD_1(x))$ 
	is a distribution over the 
	support of $\cD_0$. Then, it suffices to define:
	\[
	\cD_0'(x)= 
	\begin{cases}
	0\hspace{0.25in}& \text{if } \frac{1}{1-\alpha} \cD_0(x) \leq \cD_1(x) \\
	\frac{1}{1-(1-\alpha)(1-\delta)} (\cD_0(x) - (1-\alpha)\cD_1(x)) 
	\hspace{0.25in}  & \text{otherwise}
	\end{cases}
	\]
	and we define:
	\[
	\cD_1'(x)= 
	\begin{cases}
	0 \hspace{0.75in}& \text{if } \frac{1}{1-\alpha} \cD_0(x) \geq \cD_1(x) 
	\\
	\frac{1}{\delta} (\cD_1(x) - \frac{\cD_0(x)}{1-\alpha} ) 
	\hspace{0.75in}   & 
	\text{otherwise}
	\end{cases}
	\]
	If $\alpha = 1$, we set $\cD_0' = \cD_0$, $\cD(x) = 
	\frac{1}{1-\delta}\min (\cD_1(x), \cD_0(x))$ where $\delta$ is a scaling 
	term which ensures that $\cD(x)$ is a valid distribution. 
\end{proof}

\begin{lemma}\label{lem:clean-protocol}
	Consider any (possibly multi-round) communication protocol $\pi$ where 
	each player receives one input bit. Then for any player $i$, the  
	protocol can simulated in a manner that is ``clean'' with respect to 
	that player.
\end{lemma}

\begin{proof}
	Let $b$ denote player $i$'s bit.  We use ``round $r$'' to refer to
	the $r$th time that player $i$ is asked to speak.  Let $m_r$ be the
	transcript of the protocol just before player $i$ speaks in round
	$r$, and let $m_r^+$ denote the transcript immediately after player
	$i$ speaks in round $r$.  Let $\cD_{m_r}^b$ be the
	distribution of player $i$'s message the $r$th time he is asked to
	speak, conditioned on the transcript so far being $m_r$ and on player $i$
	having the bit $b$.  We will describe an implementation of player
	$i$ that produces outputs with the correct distribution
	$\cD_{m_r}^b$ such that the implementation only looks at $b$ with
	relatively small probability.
	
	In the first round, given $m_1$, player $i$ looks at $b$ with
	probability $\dtv(\cD_{m_1}^0, \cD_{m_1}^1)$.  If he does not look at
	the bit, he outputs each message $m$ with probability proportional
	to $\min(\cD_{m_1}^0(m), \cD_{m_1}^1(m))$; if he sees the bit $b$,
	he outputs each message $m$ with probability proportional to
	$\max(0, \cD_{m_1}^b(m) - \cD_{m_1}^{1-b}(m))$.  His output is then
	distributed according to $\cD_{m_r}^b$.  Note also that, for any
	message $m$, it is not possible that the player can send $m$ both after
	reading a $0$ and after reading a $1$.
	
	In subsequent rounds $r$, given $m_r$, player $i$ needs to output a
	message with distribution $\cD_{m_r}^b$.  Let $p_0$ denote the
	probability that the player has already observed his bit in a
	previous round, conditioned on $m_r$ and $b=0$; let $p_1$ be
	analogous for $b=1$.  We will show by induction that
	%\[
	$\min(p_0, p_1) = 0$ 
	%\]
	for all $m_r$.  That is, any given transcript may be compatible with
	having already observed a $0$ or a $1$ but not both.  As noted above,
	this is true for $r=2$.
	
	Without loss of generality, suppose $p_1 = 0$.  We apply
	Lemma~\ref{lem:distributions} to $\cD_{m_r}^0$ and $\cD_{m_r}^1$
	with $\alpha = p_0$, obtaining 
	three distributions $(\cD, \cD_0, \cD_1)$ 
	such 
	that
	%\begin{align*}
	$\cD_{m_r}^0 = (1-p_0)(1-\delta)\cD +  (1 - (1-p_0)(1-\delta))\cD_0$ and 
	$\cD_{m_r}^1 = (1-\delta)\cD + \delta \cD_1,$
	%\end{align*}
	and $\cD_0$ is disjoint from $\cD_1$.
	
	Player $i$ behaves as follows: if he has not observed his bit
	already, he does so with probability $\delta$.  After this, if he
	still has not observed his bit, he outputs a message according to
	$\cD$; if he has observed his bit $b$, he outputs according to
	$\cD_b$.
	
	The resulting distribution is $\cD_{m_r}^b$ regardless of $b$, and
	the set of possible transcripts where a $1$ has been observed is
	disjoint from those possible where a $0$ has been observed.  By
	induction, this holds for all rounds $r$. Thus, this is a simulation of 
	the original protocol that is "clean" with respect to player $i$.
\end{proof}

\begin{lemma}\label{lem:alt-protocol}
	Consider any (possibly multiround) communication protocol $\pi$
	where each player receives one bit. Each player $i$ can be
	implemented such that, if every other player receives a $0$ input,
	player $i$ only observes his input with probability
	$\dtv(\pi_{e_i}, \pi_{0})$.
\end{lemma}

\begin{proof}
	Using Lemma \ref{lem:clean-protocol}, we know that player $i$ can be 
	implemented such that the protocol is clean with respect to that player. 
	
	We may now analyze the probability $p^*$ that player $i$ ever
        observes his bit, assuming that all other players receive the
        input zero.  For every possible transcript $m$ let $p_0(m)$
        denote the probability, conditioned on the transcript being
        $m$ and player $i$'s bit being $0$, that player $i$ observes
        his bit at any point during the protocol; let $p_1(m)$ be
        analogous for the bit being $1$.  Because the choice of player
        $i$ to observe his input bit in a clean protocol is
        independent of the bit, we have that
	%\[
	$p^* = \sum_m \Pr_{\pi_0}[m] p_0(m) = \sum_m \Pr_{\pi_{e_i}}[m] p_1(m).$
	%\]
	Moreover, because the protocol is independent of the bit if it is
	not observed,
	\[
	(1-p_0(m))\Pr_{\pi_0}[m] = (1-p_1(m))\Pr_{\pi_{e_i}}[m]
	\]
	for all $m$.  By the definition of a clean protocol, the last
        message player $i$ sends can be consistent with him observing
        a $0$ or a $1$ but not both; therefore $p_0(m) = 0$ or $p_1(m) = 0$ for
        all $m$.  Now, define
        $S := \{m \mid p_0(m) > 0\} = \{m \mid \Pr_{\pi_0}(m) >
        \Pr_{\pi_{e_i}}(m)\}$.  Therefore
	\begin{align*}
	\dtv(\pi_0, \pi_{e_i}) = \sum_{m \in S} \Pr_{\pi_0}[m] - 
	\Pr_{\pi_{e_i}}[m]
	= \sum_{m \in S} p_0(m)\Pr_{\pi_0}[m]
	= p^*
	\end{align*}
	as desired.
\end{proof}

Lemma~\ref{lem:alt-protocol} will be used to show that each player has
a decent chance of not reading their input.  But to get a lower bound
for $\disj$, we need a large \emph{set} of players that have a
nontrivial chance of all ignoring their input at the same time.  We
show the existence of such a set, despite the players not being
independent.  For any $c \in (0, 1)$, define
\begin{align}\label{eq:gamma}
  \gamma_c := \frac{1}{c\log(e/c)}
\end{align}
We have

\begin{lemma}\label{lem:find-set}
Let $c \in (0,1)$, $p \in (0,\frac{1-c}{2})$, and $\gamma_c$ as in~\eqref{eq:gamma}. For a set of 0-1 random 
variables $Y_1, \dotsc, Y_k$ such that $\E[\sum_i Y_i] = pk$,
there exists $S\subset \{1,2, \dotsc, n\}$ of size $ck$ such that 
%\[
$\Pr[\forall j\in S, Y_j = 0 ] > e^{- k/\gamma_c - 1}.$
%\] 
\end{lemma}
\begin{proof}
We wish to show that there exists a set $S$ such that $Y_i = 0$ for all 
$i\in 
S$ with nontrivial probability. Observe that if $S$ were chosen uniformly at 
random,
\begin{align*}
\E_{S:\abs{S} = ck} \Pr[\forall j\in S, Y_j = 0] &\geq \frac{1}{\binom{k}{ck}} \Pr[wt(Y) \leq k-ck]
 \geq \left(\frac{c}{e}\right)^{ck} \cdot (1 - \frac{p}{1-c})
  \geq e^{-1-kc\log(e/c)}.
\end{align*}
where the first inequality considers the existence of such a set,
the second inequality uses
$\binom{a}{b} \leq (\frac{e\cdot a}{b})^b$ and Markov's
inequality, and $wt(Y)$ denotes the Hamming weight of $Y$, i.e., number of non-zero entries
of the vector $Y$. Therefore there 
exists a set $S$ of size $\Omega(ck)$ such that $\Pr[Y_S = 0] \geq 
e^{-1-kc\log(e/c)}$.
\end{proof}

We can now bound the 1-bit communication cost of our problem.
\begin{lemma}\label{lem:rand-1-index}
  Given $0 < \delta, c < 1$, $\gamma_c$ as in~\eqref{eq:gamma}, and
  $k\le \gamma_c\log(\frac{1}{2e\delta})$, for any $\delta$-error
  protocol for $F_{ck, k}$ we have that
	%\[
	$cCost_{\mu_0, \delta}(\pi) = \Omega((1-c)^2).$
	%\]
\end{lemma}

\begin{proof}
	Let $\pi$ be a protocol for $F_{ck, k}$. Let $\pi_{x}$ is the 
	distribution of the transcript of the protocol on input $x$. We start by 
	establishing a connection between conditional information cost and total 
	variation distances. First observe that due to the choice of 
	distribution $\mu_0$, we may write the conditional mutual information as:
	\begin{align*}
	cCost_{\mu_0, \delta}(\pi) = I(\pi(X_1, \dotsc, X_k); X_1, \dotsc, 
	X_k\mid D)
	 = \E_{i\in [k]} [ I(X_i; \pi_{0,0,0,\dotsc X_i,\dotsc 0,0,0})].
	\end{align*}
	Since $X_i$ is uniformly picked from $\{0,1\}$, this mutual
        information is a Jensen-Shannon divergence
        (see, for example, Wikipedia~\cite{wiki:JSD} or Proposition A.6 of \cite{BJKS}):
	\begin{align*}
	I(X_i; \pi_{0,0,0,\dotsc X_i,\dotsc 0,0,0})
	 = D_{JS}(\pi_{0}, \pi_{e_i})
	= \frac{1}{2}\left(D_{KL}(\pi_{0}, \frac{1}{2}(\pi_0+ \pi_{e_i})) + 
	D_{KL}(\pi_{e_i}, \frac{1}{2}(\pi_0 + \pi_{e_i}))\right)
	\end{align*}
	From Pinsker's inequality, $D_{KL}(P, Q) \geq 
	\frac{1}{2}\dtv^2(P, Q)$, so:
	\begin{align}
	cCost_{\mu_0, \delta}(\pi) \geq \frac{1}{4} \E_{i \in 
          [k]}[\dtv^2(\pi_0, \frac{1}{2}(\pi_{0}+\pi_{e_i})) + \dtv^2(\frac{1}{2}(\pi_{0}+\pi_{e_i}), \pi_{e_i})] = \frac{1}{8}\E_{i \in [k]}[\dtv^2(\pi_0, \pi_{e_i})].
          \label{eq:CCostTv}
	\end{align}
	
	This is similar to the connection established in Lemma 6.2 of 
	\cite{BJKS} between conditional information cost and squared Hellinger 
	distance (it is weaker but simpler to show). 
	
	Suppose, for the sake of contradiction, that $\sum_i \dtv(\pi_{e_i}, 
	\pi_{0}) = kp$ where $p < \frac{1-c}{2}$. Suppose for each player $i\in [k]$, 
	that $\dtv(\pi_{e_i}, \pi_{0}) = p_i$. By Lemma 
	\ref{lem:alt-protocol}, this implies that each player in the protocol 
	can be equivalently implemented in a manner such that---if everyone else 
	receives a $0$---player $i$ only looks at their input with 
	probability $p_i$. If a player does not look at his bit, it means
        the player's messages are
        independent of his input. 
        Let $Y_i$ denote the indicator random variable for 
	the event that player $i$ looks at his input in this equivalent 
	protocol. 
	
	For the input $X = 0$, we have $\E[\sum_i Y_i] = \sum p_i = kp$. 
	Observe, that for any set $S$, if $Y_i = 0$ for all $i\in S$, the 
	players do not see their input. So if $E_S$ denotes the event 
	that $\forall i\in S, Y_i = 0$, then
	\begin{align*}
	\dtv( \pi_{e_S}, \pi_{0}) & = \Pr[E_S]\cdot \dtv( \pi_{e_S} \mid 
	E_S, \pi_{0} \mid E_S) + \Pr[\overline{E_S}] \dtv( \pi_{e_S} \mid 
	\overline{E_S}, \pi_{0} \mid \overline{E_S}) \leq \Pr[\overline{E_S}]
	\end{align*}
	Since $\E[\sum_i Y_i] = kp$ for $p<\frac{1-c}{2}$, this means by 
	Lemma~\ref{lem:find-set} that there exists a set $S$ with 
	$\abs{S} = ck$ such that $\Pr[E_S] \geq e^{-k/\gamma_c  - 1}$. Since $k \leq 
	\gamma_c\log(\frac{1}{2e\delta})$, we have $\Pr[E_S] 
	> 2\delta$.
	For this $S$, we have that $\dtv( \pi_{e_S}, \pi_{0}) < 1 - 2\delta$ 
	and this means that the protocol errs with probability $>\delta$. This 
	is a contradiction. So, we must have $\sum_i \dtv(\pi_{e_i}, \pi_{0})
	> \frac{1-c}{2}k$. 
        By~\eqref{eq:CCostTv} and Jensen's inequality, this gives
        \[
          cCost_{\mu_0, \delta}(\pi) \geq \frac{1}{8}\E_{i\in [k]}[\dtv^2(\pi_{e_i}, \pi_{0})] \geq \frac{1}{8}\E_{i\in [k]}[\dtv(\pi_{e_i}, \pi_{0})]^2  \geq \frac{(1-c)^2}{32}.
        \]
\end{proof}

\subsection{Finishing it Off}
We prove a lower bound on the randomized communication complexity of $\disj$.
\begin{theorem}\label{thm:rand-lb}
	Given $0 < \delta, c < 1$ and $k 
	\le \gamma_c\log(\frac{1}{2e\delta})$ for $\gamma_c$ as in~\eqref{eq:gamma},
	\[ 
	R_{\delta}(\disj_{n, ck, k}) = \Omega((1-c)^2n).
	\]
\end{theorem}
To prove this, it suffices to prove Lemma 
\ref{lem:rand-1-index} where we show a lower bound on the conditional 
information cost of $\delta$-error protocols for $F_{ck, k}$. This implies a 
lower bound on the conditional information complexity of $F_{ck, k}$ which 
together with \eqref{eq:disj-dirsum} implies the desired result. 
\begin{proof}
  Combining Proposition \ref{cor:rand-cic}, Equation~\eqref{eq:disj-dirsum},
  and Lemma~\ref{lem:rand-1-index} gives:
	\begin{align*}
	R_{\delta}(\disj_{n, ck, k})
	 \geq CIC_{\eta_0, \delta}(\disj_{n, ck, k})
	\geq n \cdot CIC_{\mu_0, \delta} (F_{ck, k})
	\gtrsim n(1-c)^2
        \end{align*}
        as desired.
\end{proof}
In the Lemma \ref{lem:rand-1-index} we showed that for any protocol for $F_{ck, 
k}$ with input drawn from $\mu_0$, if the conditional information cost is 
$o(1)$, there exists an input on which it errs with probability $> \delta$. 
This implies a lower bound on the conditional information complexity of $F_{ck, 
k}$. 

%\subsection{Communication Lower Bound for $\disj$}
%Theorem \ref{thm:rand-lb} implies a more general lower bound for $\disj$ for 
%arbitrary $\delta$. 
For algorithms that have large error probability, the success 
probability can be amplified by using independent copies of the algorithm and 
taking the majority vote. We use this observation to obtain a lower bound for 
algorithms with error probability larger than $e^{-k}$. 

\restate{thm:main}

\begin{proof}
  For the absolute constant $\gamma_c$, when
  $k < \gamma_c\log(1/\delta)$ (or $\delta < e^{-k/\gamma_c}$), Theorem
  \ref{thm:rand-lb} gives us a lower bound of $\Omega(n)$. Now,
  consider the case where $\delta > e^{-k/\gamma_c}$. Suppose $\pi$ is a
  protocol whose communication cost is $C$. Then, we may amplify
  the success probability of this protocol. We create a new protocol
  $\pi'$ which runs $r$ independent copies of $\pi$ in parallel and
  outputs the majority vote across these copies.  The probability of
  failure for this new protocol is:
  $\Pr[\geq r/2 \text{ copies of } \pi \text { fail}] \leq
  \binom{r}{r/2} \delta^{r/2} \leq (4\delta)^{r/2}$.  This achieves
  failure probability $e^{-k/\gamma_c}$ for
  $r = O_c(\frac{k}{\log(1/\delta)})$. The lower bound of $\Omega_c(n)$ on
  the communication complexity of $e^{-k/\gamma_c}$-error protocols
  implies that the communication cost of $\pi$ is lower bounded by
  $\Omega(n \frac{\log(1/\delta)}{k})$ in this case.
\end{proof}

\section{Lower Bounds for $\ell_2$-Heavy Hitters}

In this section, we will prove lower bounds for certain variants of the  
$\ell_2$ heavy hitters problem in the insertion-only model. Our first lower 
bound follows from some simple observations and the lower bounds that 
follow use reductions from the Mostly Set Disjointness problem and the 
lower bound proved in the previous section.

\begin{definition}
	Given $p >1$, in the $\eps$-$\ell_p$-heavy hitters problem, we are 
	given $\eps\in (0,1)$ and a stream of items $a_1, \dotsc, a_m$ where 
	$a_i\in [n]$. If $f_i$ denotes the frequency of item $i$ in the stream, 
	the algorithm should	output all the elements $j\in [n]$ 
	such that: 
	\[
		\abs{f_j} \geq \eps \norm{p}{f}
	\]
\end{definition}

\begin{theorem}\label{thm:linear-sketching}
	Given $\eps\in (0,\frac{1}{4}]$, any deterministic linear sketching 
	algorithm for the $\eps$-$\ell_2$-heavy hitters problem must use at 
	least $\Omega(n)$ bits of space even for nonnegative vectors.	
\end{theorem}

\begin{proof}
	Assume for the sake of contradiction that $r  = o(n)$ and $M\in 
	\R^{r\times n}$ is the sketching matrix which is associated with a 
	deterministic algorithm for $1/4$-$\ell_2$ heavy hitters. We may assume 
	that $M$ has orthonormal rows (else there is an orthonormal $r\times n$ 
	matrix whose sketch is linearly related to the sketch in the algorithm 
	and we consider that matrix).
	
	Since $M$ is orthonormal we have $\sum_{i\in [n]} \norm{2}{M^TM e_i}^2 
	\leq r$. So, there must exists an $i^*\in [n]$ such that $\norm{2}{M^TM 
	e_{i^*}}^2 \leq r/n$. Consider the vector $v = e_{i^*} - M^T M e_{i^*}$ 
	which lies in the kernel of $M$. Observe that $v_{i^*}^2 \geq (1 - 
	r/n)^2\geq 1/2$ and $\norm{2}{v}^2 \leq 1$ since $I - M^T M$ is a 
	projection.	
	
	Now, let us define $w\in \R^n$ such that for all $j\neq i^*$, $w_j = 
	\abs{v_j}$ and $w_{i^*} = 0$. Observe that $w+v$ is a non-negative 
	vector and that $i^*$ is a heavy hitter in $(w+v)$ because 
	$(w+v)_{i^*}^2 \geq 1/2$ and $\norm{2}{w+v}^2 \leq (2 \norm{2}{v})^2 
	\leq 4$. Since $v$ is in the kernel of $M$ , $M(w+v) = Mw$ and the 
	algorithm must give the same output for both $(w+v)$ and $w$. However, 
	$i^*$ is a heavy hitter in $(w+v)$ and is not a heavy hitter in $w$. 
	Hence, by contradiction, $r = \Omega(n)$.
\end{proof}
In Theorem \ref{thm:delta-lb}, we prove a lower bound on the space complexity 
of $\delta$-error $r$-pass streaming algorithm for $\eps$-$\ell_p$-heavy 
hitters through a reduction from Mostly Set Disjointness.

\state{thm:delta-lb}

\begin{proof}
	Let $\mathcal{A}$ be a $\delta$-error $r$-pass streaming algorithm for 
	$\eps$-$\ell_p$-heavy hitters in the insertion-only model. We describe a 
	multiparty protocol to deterministically solve the Mostly Set 
	Disjointness problem i.e., $\disj_{n, \eps (4n)^{\frac{1}{p}}, 2\eps 
	(4n)^{\frac{1}{p}}}$ that uses the $\mathcal{A}$. The players simulate a 
	stream which updates a vector $x\in \R^{2n}$. Instead of starting with 
	$0^{2n}$ (as is the case with most streaming algorithms), the protocol 
	starts off with a vector 
	\[
	\setlength{\arraycolsep}{0pt}
	f_0 = 
	\left( 
	\begin{array}{c}
	\begin{array}{c}
	0 \\ \vdots \\ 0
	\end{array} \\
	\begin{array}{c}
	1 \\ \vdots \\ 1
	\end{array}
	\end{array}
	\right)
	~ % Some space
	\begin{array}{c}
	\noleftdelimiter
	\vphantom{\begin{array}{c}
		0 \\ \vdots \\ 0
		\end{array}}
	\right\} n \\
	\noleftdelimiter
	\vphantom{\begin{array}{c}
		x_{n_1+1} \\ \vdots \\ x_n
		\end{array}}
	\right\} n
	\end{array}
	\]
	Each player performs an update $f\leftarrow f + \delta_i$ to the vector 
	and passes the state of $\mathcal{A}$ to the next player. The update 
	vector $\delta_i$ that is processed by player $i$ is just their input 
	$x_i$ padded to length $2n$.  
	\[
	\setlength{\arraycolsep}{0pt}
	\delta = 
	\left( 
	\begin{array}{c}
	\begin{array}{c}
	\\ x_i \\ \\ 
	\end{array} \\
	\begin{array}{c}
	0 \\ \vdots \\ 0
	\end{array}
	\end{array}
	\right)
	\]
	Observe that if the input to the players is a NO-instance of  
	$\disj_{n,  \eps  (4n)^{1/p}, 2\eps(4n)^{1/p}}$, 
	then the final vector $f'$ in the turnstile stream consists of 0-1 
	entries with at least $n$ 1-s. So, $\norm{p}{f'}^p \geq n$ and since 
	$\eps\geq n^{1/p}$, no element is a $\eps$-$\ell_p$ 
	heavy hitter. 
	
	If the input is a YES-instance, then the final vector $f'$ 
	consists of $\leq 2n - 1$ entries that are 1 and one entry at which is 
	$\eps (4n)^{\frac{1}{p}}$. Since $4\eps^p n \geq \eps^p ( 2n + 
	4\eps^p n)$, that entry is a $\eps$-heavy hitter. Using the lower bound 
	of Theorem \ref{thm:main}, we know that the total communication in 
	the protocol is $\Omega(\min(n, 
	n\frac{\log(1/\delta)}{\eps n^{1/p}}))$. Since the number of messages 
	sent over $r$ rounds in the protocol is $r\cdot 2\eps(4n)^{1/p}$, there 
	exists at least one player whose communication is:
	\[
	\Omega \Big(\min\big(\frac{n^{1-\frac{1}{p}}}{r\eps}, 
	\frac{n^{1-\frac{2}{p}}\log(1/\delta)}{r\eps^2} \big) \Big)
	\]
	bits and this is a lower bound on the space complexity of $\mathcal{A}$. 
\end{proof}

A deterministic lower bound follows as a consequence of this lower bound. 
\state{thm:p-pass-lb}

In real world applications, one is concerned with lower bounds for naturally 
occurring frequency vectors. One such naturally occurring frequency 
distribution is a power law frequency distribution where the $i^\text{th}$ most 
frequent element has frequency $\propto \frac{1}{i^\zeta}$ 
where $\zeta$ typically lies in $(0.5,1]$. Formally:

\begin{definition}
Let $f\in\R^n$ be a vector such that $\abs{f_{(1)}} \geq \abs{f_{(2)}} \geq 
\dotsc \abs{f_{(n)}}$. We say that this vector is power law distributed with 
parameter $\zeta$ if for all $i\in [n]$, 
\[
	\abs{f_{(i)}} = \Theta(f_{(1)} \cdot i^{-\zeta}) + O(1) 
\]
\end{definition}

In the next theorem, we prove a lower bound on the space complexity of 
streaming algorithms for $\ell_p$-heavy hitters when the frequency vector is 
power law distributed. We denote $H_m = \sum_{i=1}^{\infty} i^{-m}$ which is 
finite when $m>1$.  
\begin{theorem}\label{thm:power-law}
  Given $p\geq 1$, $\zeta \in (\frac{1}{p}, 1]$ and $\eps \in 
  (\frac{1}{n^{\zeta}}, \frac{1}{(2 + 2\cdot H_{p\zeta})^{1/p}})$, any 
  $\delta$-error $r$-pass streaming algorithm for the $\eps$-$\ell_p$-heavy 
  hitters problem where the frequency vector is power law distributed with
  parameter $\zeta$ must have space complexity of
  $\Omega(\min(n^{1-\zeta}, n^{1-2\zeta}\log(1/\delta)))$. 
\end{theorem} 
\begin{proof}
	Let $\mathcal{A}$ be a one-pass deterministic streaming algorithm for 
	$\ell_2$ heavy hitters when the frequency vector is power law 
	distributed with parameter $\zeta$. We will use a reduction similar to 
	the Theorem \ref{thm:delta-lb} to deterministically solve $\disj_{n, 
	n^{\zeta}, 2n^{\zeta}}$ using $\mathcal{A}$.
	
	Instead of padding the initial vector $f_0$ with 1's as in Theorem 
	\ref{thm:p-pass-lb}, we pad with $\frac{2n^{\zeta}}{i^{\zeta}}$ for 
	$i\in [2, n]$. 
	
	\[
	\setlength{\arraycolsep}{0pt}
	f_0 = 
	\left( 
	\begin{array}{c}
	\begin{array}{c}
	0 \\ 0 \\ \vdots \\ 0
	\end{array} \\
	\begin{array}{c}
	2n^{\zeta}\cdot 2^{-\zeta} \\ 2n^{\zeta}\cdot 3^{-\zeta} \\ \vdots \\ 2
	\end{array}
	\end{array}
	\right)
	~ % Some space
	\begin{array}{c}
	\noleftdelimiter
	\vphantom{\begin{array}{c}
		0 \\ 0 \\ \vdots \\ 0
		\end{array}}
	\right\} n \\
	\noleftdelimiter
	\vphantom{\begin{array}{c}
		2n^{\zeta}.2^{-\zeta} \\ 2n^{\zeta}.3^{-\zeta} \\ \vdots \\ 2
		\end{array}}
	\right\} n
	\end{array}
	\]
	
	Now, suppose the players pass this frequency vector and successively 
	perform updates to obtain the final frequency vector $f'$. In the YES 
	instance, there exists one index $i\in [n]$ such that $\abs{f'_i}^p = 
	n^{p\zeta}$ and in the NO instance for all $i\in [n]$, we have 
	$\abs{f'_i} \leq 1$. In the NO case, we have $\norm{p}{f'}^p \geq 
	\sum_{i\in [2, n+1] } 2n^{p\zeta}\cdot i^{-p\zeta} \geq n^{p\zeta}$ and 
	in the YES case 
	\begin{align*}
	\norm{p}{f'}^p &= \sum_{i\in 
	[2n] } (f'_i)^p\\
		& \leq n^{p\zeta} + n + \sum_{i\in [2, n+1] }  
	2n^{p\zeta}\cdot i^{-p\zeta}\\
		& \leq n^{p\zeta} + n + H_{p\zeta} 2n^{p\zeta}\\
		& < (2 + 2H_{p\zeta}) n^{p\zeta}.
	\end{align*}
	
	So, in the YES instance, the heavy element is a 
	$\eps$-$\ell_p$-heavy hitter since $\eps^p < \frac{1}{2 (1+H_{p\zeta})}$ 
	and in the NO instance all the $\ell_p$-heavy hitters are indices in 
	$[n+1, 2n]$. Now, the final player runs the $\ell_p$-heavy hitter 
	algorithm and if any element from $[1,n]$ is a heavy hitter they output 
	YES and they output NO otherwise. 
	
	So, we have described a reduction from $\ell_p$ heavy hitters for power 
	law distributed vectors to Mostly Set Disjointness. Using Theorem 
	\ref{thm:main}, the total communication here is lower bounded by 
	$\Omega( n, n^{1-\zeta}\log(1/\delta))$. Since there are $n^{\zeta}$ 
	players, the space complexity lower bound for the streaming algorithm is 
	$\Omega( n^{1-\zeta}, n^{1-2\zeta}\log(1/\delta))$.
\end{proof}

\section{Application to Low Rank Approximation}\label{sec:lowrank}

As an application of our deterministic $\ell_2$-heavy hitters lower bound in 
insertion streams, we prove a lower bound for the low rank approximation 
problem in the standard row-arrival model in insertion streams: we see rows 
$A_1, A_2, \ldots, A_n$ each in $\mathbb{R}^d$, one at a time. At the end of 
the stream we should output a projection $P$ onto a rank-$k$ space for which 
$\|A-AP\|_F^2 \leq (1+\epsilon)\|A-A_k\|_F^2$, where $A_k$ is the best rank-$k$ 
approximation to $A$. The {\sf FrequentDirections} algorithm provides a 
deterministic upper bound of $O(dk/\epsilon)$ words of memory (assuming entries 
of $A$ are $O(\log(nd))$ bits and a word is $O(\log(nd))$ bits) was shown in 
\cite{l13,gp14}, and a matching lower bound of $\Omega(dk/\epsilon)$ words of 
memory was shown in \cite{w14}. See also \cite{g16} where the upper and lower 
bounds were combined and additional results for deterministic algorithms were 
shown.

A natural question is if {\sf FrequentDirections} can be improved when the rows 
of your matrix are sparse. Indeed, the sparse setting was shown to have faster 
running times in \cite{glp16,h18}. Assuming there are $n$ rows and each row has 
$s$ non-zero entries, the running time was shown to be $O(sn(k +\log n) + nk^3 
+ d(k/\epsilon)^3)$, significantly improving the $nd$ time required for dense 
matrices. Another question is if one can improve the memory required in the 
sparse setting. The above lower bound has an $\Omega(d)$ term in its complexity 
because of the need to store directions in $\mathbb{R}^d$. However, it is 
well-known \cite{gs12} that any matrix $A$ contains $O(k/\epsilon)$ rows whose 
row-span contains a rank-$k$ projection $P$ for which $\|A-AP\|_F^2 \leq 
(1+\epsilon)\|A-A_k\|_F^2$. Consequently, it is conceivable in the stream one 
could use $O(s k / \epsilon)$ words of memory in the sparse setting, which 
would be a significant improvement if $s \ll d$. Indeed, in the related 
communication setting, this was shown to be possible in \cite{bwz16}, whereby 
assuming the rows have at most $s$ non-zero entries it is possible to find such 
a $P$ with communication only $O(sk/\epsilon)$ words per server, improving upon 
the $O(dk/\epsilon)$ words per server bound for general protocols, at least in 
the randomized case. It was left open if the analogous improvement was possible 
in the streaming setting, even for deterministic algorithms such as {\sf 
FrequentDirections}.

Here we use our deterministic lower bound to show it is not possible to remove 
a polynomial dependence on $d$ in the memory required in streaming setting for 
deterministic algorithms.

\state{thm:rank}

%\begin{theorem}\label{thm:detLB}
%Any single pass deterministic streaming algorithm outputs a rank-$k$ projection 
%matrix $P$ must use $\Omega(\sqrt{d})$ bits of memory, even for $k = 1$, 
%$\epsilon = \Theta(1)$, and if each row of $A$ has only a single non-zero entry.
%\end{theorem}
\begin{proof}
Recall in one instantiation of our hard communication problem, the players have 
sets $S_1, \ldots, S_{\sqrt{d}} \subseteq \{1, 2, \ldots, d\}$ each of size 
$\sqrt{d}/2$ and either the sets are pairwise disjoint or there exists a unique 
element $i^*$ occurring in at least $2/3$ fraction of the sets. We associate 
each element $i$ in each set $S_{\ell}$ with a row of $A$ which the standard 
unit vector $e_i$ which is $1$ in position $i$ and $0$ in all remaining 
positions. The stream is defined by seeing all the rows corresponding to 
elements in $S_1$, then in $S_2$, and so on.

Suppose we have seen the first $1/2$ fraction of sets in the stream. In this 
case, the row $i^*$ must have occurred in at least $1/2-1/3 = 1/6$ fraction of 
sets. Thus, at this point in the stream, the top singular value of $A$ is 
$\sqrt{d}/6$ and all remaining singular values of $A$ equal $1$. Now, the 
algorithm outputs a rank-$1$ projection $P$ from its internal memory state. 
Suppose $P = vv^T$ for a unit vector $v$. Then
$$\|A-Avv^T\|_F^2 = \|A\|_F^2 - \|Av\|_2^2 \geq \|A\|_F^2 - 1 + (d/36) 
v_{i^*}^2.$$
Consequently, to obtain a $C$-approximation for a sufficiently small constant 
$C > 1$, we must have $v_{i^*}^2 = \Omega(1)$. Since $\|v\|_2^2 = 1$, there is 
a set $T$ of size $O(1)$ which contains all indices $j$ for which $v_j^2 = 
\Omega(1)$.

Now, since we have only observed a $1/2$ fraction of sets in the stream, the 
element $i^*$ must occur in at least $2/3-1/2 = 1/6$ fraction of sets in the 
remaining half of the stream. Thus, for each element in the set $T$, we can 
check if it occurs at all in the second half of the stream. However, if there 
is such an element $i^*$, it must be the only element in $T$ occurring in the 
second half of the stream. In case the sets in our hard instance are pairwise 
disjoint, no element in $T$ will occur in the second half of the stream. Thus, 
we can deterministically distinguish which of the two cases we are in.

Note that the maximum communication of this reduction is the memory size of the 
streaming algorithm, together with an additional additive $O(\log d)$ bits of 
memory to store $T$. Thus, we get that the memory required of our streaming 
algorithm is at least $\Omega(\sqrt{d}) - O(\log d) = \Omega(\sqrt{d})$ bits.
\end{proof}

\section{Algorithm for bounded-length turnstile streams}

In this section we show that $\ell_2$ heavy hitters on turnstile
streams of length $O(n)$ can be solved in $O(n^{2/3})$ space.  This is
intermediate between the $O(\sqrt{n})$ possible in the insertion-only
model and the $\Omega(n)$ necessary in linear sketching.

\state{thm:alg}
\begin{proof}
  Let $S$ be a parameter to be determined later.  We run three
  algorithms in parallel: space-$O(S)$ \textsc{FrequentElements} on the
  positive updates to $x$; space-$O(S)$ \textsc{FrequentElements} on the
  negative updates to $x$ (with sign flipped to be positive); and a
  linear sketching algorithm for exact $S$-sparse recovery (e.g.,
  Reed-Solomon syndrome decoding).

  Let $P,N$ be the number of positive/negative updates, respectively,
  so $L = P+N$.  Let $x^+$ and $x^-$ be the sum of positive/negative
  updates, so $x = x^+ - x^-$.  The two \textsc{FrequentElements}
  sketches give us estimates $\wh{x}^+$ and $\wh{x}^-$, respectively,
  such that for each $i$:
  \begin{align*}
    x^+_i - P/S &\leq \wh{x}^+_i \leq x^+_i &
    x^-_i - N/S &\leq \wh{x}^-_i \leq x^-_i
  \end{align*}
  Therefore $\wh{x} := \wh{x}^+ - \wh{x}^-$ satisfies
  \[
    \|\wh{x} - x\|_\infty \leq \max(P/S, N/S) \leq L/S.
  \]
  Second, the $S$-sparse recovery algorithm gives us a $\wh{y}$ such
  that, if $\|x\|_0 \leq S$, $\wh{y}_i = x_i$ for all $i$.

  For a strict turnstile stream, we can compute $\|x\|_1 = P-N$.  Our
  algorithm outputs the $\eps$-heavy hitters of $\wh{y}$ if
  $\|x\|_1 \leq S$, and otherwise outputs the entries of
  $\wh{x}$ larger than $3L/S$.

  Since $\|x\|_0 \leq \|x\|_1$, the output is exactly correct when
  $\|x\|_1 \leq S$.  Otherwise, $\|x\|_2 \geq \sqrt{\|x\|_1} \geq \sqrt{S}$, so 
  for $S \geq (L/\eps)^{2/3}$,
  \[
    \|\wh{x} - x\|_\infty \leq L/S \leq \eps \sqrt{S} \leq \eps \|x\|_2.
  \]
  Therefore the algorithm will output all $4\eps$-heavy hitters and
  only $2\eps$-heavy hitters.  Rescaling $\eps$ by 4 gives the
  standard $\ell_2$ heavy hitters guarantee.
\end{proof}

\begin{remark}\label{remark:turn}
  For non-strict turnstile streams, one can still achieve the
  $\ell_\infty/\ell_2$ guarantee
  \[
    \|\wh{z} - x\|_\infty \leq \eps \|x\|_2
  \]
  with the same space, but the $\ell_2$ heavy hitters guarantee (of
  outputting all $\eps$-heavy hitters and only $\eps/2$-heavy hitters)
  requires $\Omega(\min(n, L))$ space.
\end{remark}
\begin{proof}
  To achieve the $\ell_\infty/\ell_2$ guarantee, we combine $\wh{x}$
  and $\wh{y}$ in the above algorithm slightly differently: if
  $\|\wh{y} - \wh{x}\|_\infty \leq L/S$, output $\wh{y}$; else, output
  $\wh{x}$.  Call this output $\wh{z}$.  We have that
  $\|\wh{z} - x\|_\infty \leq \|\wh{z} - \wh{x}\|_\infty + \|\wh{x} -
  x\|_\infty \leq 2L/S$ unconditionally, and $\wh{z} = x$ if
  $\|x\|_0 \leq S$.  The algorithm outputs $\wh{z}$.

  So when $\|x\|_0 \leq S$, this algorithm recovers $x$ exactly and
  certainly finds the heavy hitters.  On the other hand, when
  $\|x\|_0 \geq S$, we have $\|x\|_2 \geq \sqrt{S}$.  Therefore for
  $S \geq 2(L/\eps)^{2/3}$,
  \[
    \|\wh{z} - x\|_\infty \leq 2L/S \leq \eps \sqrt{S} \leq \eps \|x\|_2
  \]
  as desired.

  For the lower bound, it suffices to consider $L = \Theta(n)$
  [otherwise, restrict to the first $\Theta(L)$ coordinates/do nothing
  interesting after the first $O(n)$ updates].  We can solve
  \textsc{Equality} on $b = n/10$ bits as follows: using a
  constant-distance, constant-rate code, associate each input
  $y \in \{0, 1\}^b$ with a codeword $C_y \in \{0, 1\}^{n-1}$, such
  that $\|C_y - C_{y'}\|_1 > n/10$ for all $y \neq y'$.  Alice, given
  the input $y$, inserts $x_1 = 1$, then inserts $C_y$ on the
  remaining coordinates.  She sends the sketch of the result to Bob,
  who subtracts his $C_{y'}$ from coordinates $2, \dotsc, n$ and asks
  for the $\eps$-heavy hitters of the result.  For any
  $1 > \eps > 10/\sqrt{n}$, this list will contain coordinate $1$ if
  and only if $y = y'$, solving equality, giving the desired
  $\Omega(n)$ bound.  [And since $\eps$-heavy hitters exactly
  reconstructs binary vectors on $1/\eps^2$ coordinates, an
  $\Omega(n)$ bound for $\eps \leq O(1/\sqrt{n})$ is trivial.]
\end{proof}

\section*{Acknowledgments}
The authors would like to thank the anonymous reviewers of a previous
version of this paper for helpful suggestions that significantly
improved the presentation. D. Woodruff would like to thank support
from NSF grant No. CCF-1815840 and a Simons Investigator Award.
E. Price would like to thank support from NSF grant CCF-1751040
(CAREER).

\bibliographystyle{alpha}
\bibliography{ref}

\end{document}